\documentclass[a4paper,reqno]{amsart}

\newtheorem{definition}{Definition}
\usepackage{graphicx}
\usepackage{xcolor}
\usepackage{tikz}
\usepackage{tikz-cd}
\usepackage{caption}
\usepackage{booktabs}
\usepackage{adjustbox}
\usepackage{longtable}
\usepackage{lscape}
\usepackage{array}
\newcommand{\be}{\begin{equation}}
\newcommand{\ee}{\end{equation}}

\usepackage{enumerate}
\usepackage{enumitem}

\newtheorem{thm}{Theorem}[section]
\newtheorem{prop}[thm]{Proposition}

\newtheorem{remark}[thm]{\it Remark}

\usepackage{tikz-3dplot}
\usepackage{xifthen}
\usepackage{caption}
\usepackage{umoline}

\tdplotsetmaincoords{60}{125}
\tdplotsetrotatedcoords{0}{20}{0} 

\usepackage[
labelfont=sf,
hypcap=false,
format=hang
]{caption}

\usepackage{mathtools}

\begin{document}


\title{Matrix factorizations and pentagon  maps}

\author[P. Kassotakis]{Pavlos Kassotakis}
\address{Pavlos Kassotakis, Department of Mathematical Methods in Physics, Faculty of Physics,
University of Warsaw, Pasteura 5, 02-093, Warsaw, Poland}
 \email{Pavlos.Kassotakis@fuw.edu.pl, pavlos1978@gmail.com}




\begin{abstract}
We propose a specific class of matrices which participate in factorization problems that turn to be equivalent to constant and entwining (non-constant) pentagon, reverse-pentagon or  Yang-Baxter  maps, expressed in  non-commutative variables. In detail,  we show that factorizations of order $N=~2$ matrices of this specific class are equivalent to the  {\em homogeneous normalization map}.
  From order $N=3$ matrices, we obtain an extension of the  homogeneous normalization map, as well as novel entwining  pentagon, reverse-pentagon and Yang-Baxter maps.
\end{abstract}

\setcounter{tocdepth}{2}


\maketitle
\section{Introduction}
The interplay of matrix factorization problems and discrete integrable systems respectively in one, two or three independent variables has been introduced and studied in various seminal papers, see respectively \cite{Symes:1982,Deift:1989,Moser:1991}, \cite{Nijhoff:1989,Adler:1994}  or \cite{Maillet:1989,korepanov:1995}. Furthermore, there exists an intrinsic relation between matrix factorization problems and discrete integrability and  it stands as an active area of research as the recent developments \cite{Veselov:20031,Veselov:2003b,koul-2009,Doliwa:2020,Rizos:2022} suggest. In this paper, we study a particular class of matrices that participate in specific matrix re-factorization problems. These re-factorization problems turn out to be equivalent to pentagon, reverse-pentagon, or Yang-Baxter maps.


 Pentagon maps serve as set theoretical solutions of the pentagon equation
  \begin{align}\label{pent}
 S^{(1)}_{12}S^{(3)}_{13}S^{(5)}_{23}=S^{(4)}_{23}S^{(2)}_{12},
 \end{align}
  that is the first non-trivial example of the so-called polygon equations c.f. \cite{Dimakis:2015,korepanov:2022}. In relation (\ref{pent}), the superscripts denote objects ($S^{(q)}$) that might differ f.i. these objects might be operators or maps; so we have respectively the ``operator" and the ``set-theoretic" version of (\ref{pent}). In the set-theoretic version of (\ref{pent}) the subscripts denote the sets where the maps $S^{(q)}$ act. For detailed interpretation of (\ref{pent}) and for the definition of the reverse-pentagon and the Yang-Baxter equation, see Section \ref{sec:2}.

    The operator form of (\ref{pent}) first appeared in \cite{MoorSeib:89} in relation with conformal field theory. In \cite{Maillet:1990} it was considered inside the context of three-dimensional integrable systems, while the set-theoretical version of  (\ref{pent}) was studied in \cite{Zakrzewski:1992} in connection with Poisson maps, see also \cite{Skandalis:1993}. Furthermore, the pentagon relation (\ref{pent}) itself, as it was shown in \cite{Korepanov:2000}, serves as  a manifestation of the $3\leftrightarrow 2$ Pachner move \cite{Pachner:1991}, where three tetrahedra with a common edge are replaced by two tetrahedra with a common face on a triangulation of a piecewise-linear 3-manifold.
  Pentagon maps appeared in the context of Roger's dilogarithm \cite{Faddeev:1994,Kashaev:1999} and are also related to the closure relation of the Lagrangian multiform theory in the setting of discrete integrable systems \cite{Lobb:2009,JNF}, as well as to cluster algebras \cite{Volkov:2011,Fordy_Hone:2014}.  For further connections and interrelations of the pentagon equation and maps with various areas of Mathematics and Physics we refer to  \cite{Dimakis:2015}.   While for recent developments on pentagon maps we refer to \cite{Zakrzewski:1992,Skandalis:1993,Kashaev:1998,Kashaev:1999,Doliwa:2014p,Doliwa:2020,Catino:2020,Colazzo:2020}.

This article is organized as follows. We begin  with a brief introduction followed by Section \ref{sec:2} where we present the  definitions and the basic mathematical notions used in this article. In  Section \ref{sec:3}, we introduce a special class of matrices that participate in matrix factorization problems that turn equivalent to pentagon, reverse-pentagon and  Yang-Baxter  maps.
 In particular, via these matrix re-factorization problems, when the order $N$ of the matrices is $N=2,$  we recover a well known pentagon map that is the homogeneous normalization map \cite{Doliwa:2014p,Doliwa:2020}.  When $N=3$, we obtain an extension of the homogeneous normalization map, along with novel entwining Yang-Baxter and pentagon maps.
 Note that  the proof that  the obtained mappings are pentagon maps, follows from direct computation. 
 Alternatively, since  the {\em local $(N-1)-$gon equation} determines an $N-$gon map   \cite{Dimakis:2015},  proof of an anticipated result that the class of matrices introduced here, satisfies
the {\em local tetragon ($2-$gon) equation},  would serve as an alternative proof of the pentagonal property.
 In appendix \ref{app1}, we explicitly provide the  tetrahedron and 4-simplex  maps associated with the presented extension of the normalization map. In the concluding Section  \ref{sec:4}, we summarize our results and we briefly discuss  the case where the order of the matrices that participate in factorizations is $N=4$.  Finally note that all mappings presented in this article  are expressed in totally non-commutative variables; that is no commutation relation is assumed among them.

\section{Preliminaries on pentagon maps and more} \label{sec:2}

Let $\mathbb{X}$ be a set. We proceed with the following definitions.

\begin{definition}
The  maps $R: \mathbb{X} \times \mathbb{X} \rightarrow \mathbb{X} \times \mathbb{X}$ and $\widehat R: \mathbb{X} \times \mathbb{X} \rightarrow \mathbb{X} \times \mathbb{X}$  will be called $M\ddot{o}b$ equivalent if it exists a bijection $\kappa: \mathbb{X} \rightarrow \mathbb{X}$ such that $(\kappa\times \kappa) \, R=\widehat R \, (\kappa\times \kappa).$
\end{definition}

Let $R^{(q)},S^{(q)},\bar S^{(q)}: \mathbb{X} \times \mathbb{X}\ni(x,y)\mapsto (u,v)=(u(x,y),v(x,y))\in \mathbb{X}\times \mathbb{X},$ $q\in\{1,\ldots,5\},$ be  maps and let  $R^{(q)}_{ij}$ $i<j\in\{1,2,3\},$ denote the maps that act as  $R^{(q)}$ on the $i-$th and $j-$th factor of $\mathbb{X} \times \mathbb{X}\times \mathbb{X},$ i.e.
$R^{(q)}_{12}:(x,y,z)\mapsto (u(x,y),v(x,y),z),$ $R^{(q)}_{13}:(x,y,z)\mapsto (u(x,z),y,v(x,z)),$ and $R^{(q)}_{23}:(x,y,z)\mapsto (x,u(y,z),v(y,z))$ and similarly define the maps $S^{(q)}_{ij},\bar S^{(q)}_{ij}.$ Then we have the respective definitions of entwining (non-constant) {\em Yang-Baxter}, {\em pentagon}, {\em reverse pentagon}, {\em tetrahedron} and {\em 4-simplex} maps. Remember that the superscripts denote maps that might differ.

\begin{definition}
The maps $R^{(q)}: \mathbb{X} \times \mathbb{X}\rightarrow  \mathbb{X} \times \mathbb{X},$
  will be called {\em  entwining (non-constant) Yang-Baxter maps} 
   if they satisfy the {\em  entwining  (non-constant) Yang-Baxter equation} \cite{Yang:1967,Baxter:1982} (see Figure \ref{fig1})
\begin{align*}
R^{(1)}_{12}\circ R^{(2)}_{13}\circ R^{(3)}_{23}= R^{(3)}_{23}\circ R^{(2)}_{13}\circ R^{(1)}_{12}.
\end{align*}
The maps $S^{(q)}: \mathbb{X} \times \mathbb{X}\rightarrow  \mathbb{X} \times \mathbb{X},$
 will be called {\em entwining (non-constant) pentagon maps} 
 if they satisfy the {\em  entwining (non-constant) pentagon equation} \cite{Biedenharn_Louck:81,MoorSeib:89} (see Figure \ref{fig2})
\begin{align*}
S^{(1)}_{12}\circ S^{(3)}_{13}\circ S^{(5)}_{23}= S^{(4)}_{23}\circ S^{(2)}_{12}.
\end{align*}
The maps $\bar S^{(q)}: \mathbb{X} \times \mathbb{X}\rightarrow  \mathbb{X} \times \mathbb{X},$
 will be called {\em entwining (non-constant) reverse-pentagon map} if they satisfy the {\em entwining (non-constant) reverse-pentagon equation}
\begin{align*}
\bar S^{(5)}_{23}\circ \bar S^{(3)}_{13}\circ \bar S^{(1)}_{12}= \bar S^{(2)}_{12}\circ \bar S^{(4)}_{23}.
\end{align*}
The maps $T^{(q)}: \mathbb{X} \times \mathbb{X}\times \mathbb{X}\rightarrow  \mathbb{X} \times \mathbb{X}\times \mathbb{X},$
 will be called {\em entwining (non-constant) tetrahedron maps}  if they satisfy the {\em entwining (non-constant) tetrahedron equation} \cite{Zamolodchikov:1980}
\begin{align*}
T^{(1)}_{123}\circ T^{(2)}_{145}\circ T^{(3)}_{246}\circ T^{(4)}_{356}=T^{(4)}_{356}\circ T^{(3)}_{246}\circ T^{(2)}_{145}\circ T^{(1)}_{123}.
\end{align*}
Alternatively, tetrahedron maps are called 3-simplex or  Zamolodchikov maps.\\
The maps $P^{(q)}: \mathbb{X}^4 \rightarrow  \mathbb{X}^4 ,$ where $\mathbb{X}^4=\mathbb{X} \times \mathbb{X}\times \mathbb{X}\times \mathbb{X}$
 will be called {\em entwining (non-constant) 4-simplex maps}  if they satisfy the {\em entwining (non-constant) 4-simplex equation} \cite{Bazh_Stro:1982}
\begin{align*}
P^{(1)}_{0123}\circ P^{(2)}_{0456}\circ P^{(3)}_{1478}\circ P^{(4)}_{2579}\circ P^{(5)}_{3689}=P^{(5)}_{3689}\circ P^{(4)}_{2579}\circ P^{(3)}_{1478}\circ P^{(2)}_{0456}\circ P^{(1)}_{0123}.
\end{align*}
Alternatively, 4-simplex maps are called pentachoron (=5-cell) or Bazhanov-Stroganov maps.
\end{definition}

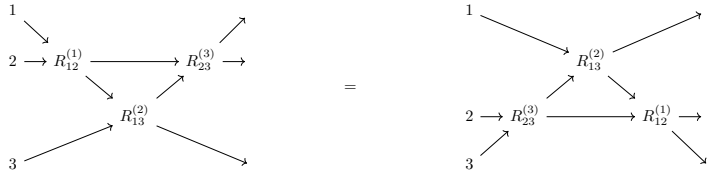
\begin{figure}[h]  \adjustbox{scale=0.6,center}{
\begin{minipage}[htb]{0.2\textwidth}
\begin{tikzcd}[row sep=0.5cm, column sep = 0.5cm,every arrow/.append style={dash}]
 1 \arrow[rightarrow,from=1-1,to=2-2]& {} & {} & {} &{} \\
 2 \arrow[rightarrow,from=2-1,to=2-2] & R_{12}^{(1)}\arrow[rightarrow,from=2-2,to=2-4] \arrow[rightarrow,from=2-2,to=3-3]  & {}& R_{23}^{(3)} \arrow[rightarrow,from=2-4,to=2-5] \arrow[rightarrow,from=2-4,to=1-5]& {} \\
  {} & {} & R_{13}^{(2)} \arrow[rightarrow,from=3-3,to=4-5] \arrow[rightarrow,from=3-3,to=2-4] &{}  & {} \\
 3 \arrow[rightarrow,from=4-1,to=3-3] &{} &{} &{} &{}
\end{tikzcd}
\end{minipage} \hspace{4cm} = \hspace{2cm}
\begin{minipage}[htb]{0.2\textwidth}
\begin{tikzcd}[row sep=0.5cm, column sep = 0.5cm,every arrow/.append style={dash}]
 1 \arrow[rightarrow,from=1-1,to=2-3]& {} & {} & {} &{} \\
  {} & {} & R_{13}^{(2)} \arrow[rightarrow,from=2-3,to=1-5] \arrow[rightarrow,from=2-3,to=3-4] &{}  & {} \\
   2 \arrow[rightarrow,from=3-1,to=3-2] & R_{23}^{(3)}\arrow[rightarrow,from=3-2,to=2-3] \arrow[rightarrow,from=3-2,to=3-4] & {} & R_{12}^{(1)} \arrow[rightarrow,from=3-4,to=3-5] \arrow[rightarrow,from=3-4,to=4-5] &{} \\
   3 \arrow[rightarrow,from=4-1,to=3-2] &{} &{} &{} &{}
\end{tikzcd}
\end{minipage}}
\caption{The Yang-Baxter relation} \label{fig1}
\end{figure}

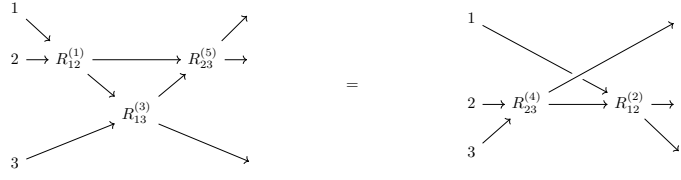
\begin{figure}[h]\adjustbox{scale=0.6,center}{
\begin{minipage}[htb]{0.2\textwidth}
\begin{tikzcd}[row sep=0.5cm, column sep = 0.5cm,every arrow/.append style={dash}]
 1 \arrow[rightarrow,from=1-1,to=2-2]& {} & {} & {} &{} \\
 2 \arrow[rightarrow,from=2-1,to=2-2] & R_{12}^{(1)}\arrow[rightarrow,from=2-2,to=2-4] \arrow[rightarrow,from=2-2,to=3-3]  & {}& R_{23}^{(5)} \arrow[rightarrow,from=2-4,to=2-5] \arrow[rightarrow,from=2-4,to=1-5]& {} \\
  {} & {} & R_{13}^{(3)} \arrow[rightarrow,from=3-3,to=4-5] \arrow[rightarrow,from=3-3,to=2-4] &{}  & {} \\
 3 \arrow[rightarrow,from=4-1,to=3-3] &{} &{} &{} &{}
\end{tikzcd}
\end{minipage} \hspace{4cm} = \hspace{2cm}
\begin{minipage}[htb]{0.2\textwidth}
\begin{tikzcd}[row sep=0.5cm, column sep = 0.5cm,every arrow/.append style={dash}]
 1 \arrow[rightarrow,from=1-1,to=3-4]& {} & {} & {} &{} \\
  {} & {} &  {}  &{}  & {} \\
   2 \arrow[rightarrow,from=3-1,to=3-2] & R_{23}^{(4)}\arrow[rightarrow, crossing over,from=3-2,to=1-5] \arrow[rightarrow,from=3-2,to=3-4] & {} & R_{12}^{(2)} \arrow[rightarrow,from=3-4,to=3-5] \arrow[rightarrow,from=3-4,to=4-5] &{} \\
   3 \arrow[rightarrow,from=4-1,to=3-2] &{} &{} &{} &{}
\end{tikzcd}
\end{minipage}}
\caption{The pentagon relation} \label{fig2}
\end{figure}
In the definitions above, when the maps $R^{(q)}$ do not differ so  we have
\begin{align*}
R^{(1)}=R^{(2)}=\ldots=R^{(q)}=R,
 \end{align*}
  we say that $R$ serves as a solution of the {\em constant} Yang-Baxter equation or just the Yang-Baxter equation and the map $R$ will be referred as {\em constant Yang-Baxter map} or simply Yang-Baxter map.  In a similar manner,  the constant pentagon, reverse-pentagon, tetrahedron and 4-simplex maps are defined. The term non-constant indicates the cases where  at least one of the maps $R^{(q)}$  differs. For  studies on entwining Yang-Baxter maps see \cite{Kouloukas:2011,Kassotakis:2019,Kels:2019II,Adamopoulou:2021}, at the same time for recent investigations on constant Yang-Baxter maps see \cite{Atkinson:2012i,Atkinson:2014,Mi+Pa+Wa:16,Kass2,Kels:2019II,Talalaev:2022,Caudrelier:2022}, whilst for non-commutative Yang-Baxter maps we refer to \cite{Doliwa_2014,Grahovski:2016,Kassotakis:2:2021,Kassotakis:2022b}.  Additionally, for recent developments on tetrahedron and 4-simplex maps we refer to \cite{Kass1,Doliwa:2020,Igonin:2022,Rizos:2022}, while for permutation type maps that serve as solutions to higher simplex equations see \cite{Hietarinta:1997}.

\begin{remark} \label{remark:1}
Let $\tau:\mathbb{X} \times \mathbb{X}\mapsto \mathbb{X} \times \mathbb{X}$ the transposition map i.e. $\tau: (x,y)\mapsto (y,x).$  For   pentagon and reverse-pentagon maps it is easy to verify the following:
\begin{itemize}
\item a map $S: \mathbb{X} \times \mathbb{X} \rightarrow \mathbb{X} \times \mathbb{X},$ is a pentagon map if and only if $\tau\circ S\circ \tau$ is a reverse-pentagon map;
\item an invertible map     $S: \mathbb{X} \times \mathbb{X} \rightarrow \mathbb{X} \times \mathbb{X},$ is a pentagon map if and only if $S^{-1}$ is a reverse-pentagon map;
\item an invertible map     $S: \mathbb{X} \times \mathbb{X} \rightarrow \mathbb{X} \times \mathbb{X},$ is a pentagon map if and only if    $\tau\circ S^{-1}\circ \tau$ is a pentagon map. The map $\tau\circ S^{-1}\circ \tau$ is referred to as the {\em opposite} pentagon map and it is denoted as $S^{op}.$  Similar statements hold for reverse-pentagon maps.
\end{itemize}
\end{remark}

\begin{prop} \label{boundary}
Let $S^{(q)}: \mathbb{X} \times \mathbb{X}\rightarrow \mathbb{X} \times \mathbb{X},$ $q\in\{1,\ldots,5\},$ be entwining (non-constant) pentagon maps and let be the maps $\chi,\psi:\mathbb{X}\rightarrow \mathbb{X}.$  Then the maps
\begin{align*}
Q^{(1)}:=(\chi\times id)\circ S^{(1)},&&Q^{(2)}:=(\chi\times id)\circ S^{(2)},&&Q^{(4)}:=S^{(4)}\circ (id\times \psi),&&Q^{(5)}:=S^{(5)}\circ (id\times \psi),
\end{align*}
satisfy the following entwining (non-constant) pentagon equation
\begin{align} \label{entw:00}
Q^{(1)}_{12}\circ S^{(3)}_{13}\circ Q^{(5)}_{23}=Q^{(4)}_{23}\circ Q^{(2)}_{12}.
\end{align}
Mappings $\chi$ and $\psi,$ will be called boundary maps.

A similar statement holds for the entwining (non-constant) reverse-pentagon maps
\begin{align*}
\bar S^{(q)}: \mathbb{X} \times \mathbb{X}\rightarrow \mathbb{X} \times \mathbb{X},&& q\in\{1,\ldots,5\}.
\end{align*}
\end{prop}
\begin{proof}
Since $S^{(q)}$ are non-constant pentagon maps, we have
\begin{align*}
Q^{(1)}_{12}\circ S^{(3)}_{13}\circ Q^{(5)}_{23}=\chi_1\circ S^{(1)}_{12}\circ S^{(3)}_{13}\circ S^{(5)}_{23}\circ \psi_3=
\chi_1\circ S^{(4)}_{23}\circ S^{(2)}_{12}\circ \psi_3.
\end{align*}
But there is
\begin{align*}
\chi_1\circ S^{(4)}_{23}\circ S^{(2)}_{12}\circ \psi_3=S^{(4)}_{23}\circ \psi_3\circ \chi_1\circ S^{(2)}_{12}=Q^{(4)}_{23}\circ Q^{(2)}_{12},
\end{align*}
and that proves that relation (\ref{entw:00}) holds.
\end{proof}
\begin{definition}[Ten-term relation \cite{Kashaev:1998}] \label{Def:ten-term}
The maps $W^q$ and $\bar W^q$ that map $\mathbb{X} \times \mathbb{X}$ to itself, will be said to satisfy the {\em non-constant or entwining ten-term} relation if they satisfy
\begin{align} \label{ten-term}
\bar W^{(1)}_{13}\circ W^{(2)}_{12}\circ \bar W^{(3)}_{14}\circ W^{(4)}_{34}\circ \bar W^{(5)}_{24}=
W^{(5)}_{34}\circ \bar W^{(4)}_{24}\circ W^{(3)}_{14}\circ \bar W^{(2)}_{13}\circ W^{(1)}_{12}.
\end{align}
\end{definition}
The ten-term relation in its entwining form (\ref{ten-term}) appeared in \cite{Dimakis_Korepanov:2021}.
\begin{prop}[\cite{Kashaev:1998}] \label{Def:Tetra-Penta}
Let $S^{(q)}$ respectively $\bar S^{(q)}$ be a pentagon respectively a reverse-pentagon map that satisfy the {\em ten-term} relation (\ref{ten-term}), then the maps
\begin{align*}
T^{(q)}_{123}:=\bar S^{(q)}_{13}\circ\tau_{23}\circ S^{(q)}_{13},
\end{align*}
define non-constant  tetrahedron maps, whereas the maps
\begin{align*}
P^{(q)}_{0123}:=\bar S^{(q)}_{13}\circ\tau_{01}\circ \tau_{23}\circ S^{(q)}_{13},
\end{align*}
define non-constant a 4-simplex maps, where $\tau:\mathbb{X} \times \mathbb{X}\rightarrow  \mathbb{X} \times \mathbb{X}$ the transposition map i.e. $\tau_{ij}$ stands for the transposition of the $i-$th and $j-$th arguments of $\mathbb{X}^N,$ $N\in\mathbb{N}$.
\end{prop}

Let $\mathbb{A}$ be an  associative algebra over a field $\mathbb{F},$  with multiplicative identity that we denote with $1.$
Throughout this paper we consider $\mathbb{X}=\underbrace{\mathbb{A}^\times\times\cdots \times \mathbb{A}^\times}_{N-\text{times}},$ $N\in \mathbb{N},$ where $\mathbb{A}^\times$ denotes the subgroup of elements $w\in \mathbb{A}$ having multiplicative inverse $w^{-1}\in \mathbb{A},$ s.t. $w w^{-1}=w^{-1} w=1.$
In this general setting, $\mathbb{A}^\times$ could be a division ring for instance bi-quaternions.  More generally, $\mathbb{A}^\times$ could stand for the subgroup of invertible matrices of the algebra $\mathbb{A}$ of $N\times N$  matrices.

\section{Matrix factorizations      pentagon and reverse-pentagon maps} \label{sec:3}

Let $I_{N}$ the order $N$ identity matrix and let $p_k,$ $k\in \{1,\ldots,K\},$
be a set of $k$ distinct positive integers $l_j$ i.e.
$$
p_k=\{l_1,\ldots, l_k\},\quad k\in \{1,\ldots,K\}.
$$
Let also $q_\lambda,$ $\lambda\in \{1,\ldots,K\},$ be  a set of $\lambda$ distinct positive integers $m_j$ i.e.
$$
q_\lambda=\{m_1,\ldots, m_\lambda\}.
$$
 Denote with $V_j(x),$ $j\in p_k$  the $N-$component row vectors
$$
V_j(x):=(x^{l_j,1},x^{l_j,2},\ldots,x^{l_j,N}),
$$
with $x^{l_j,1},x^{l_j,2},\ldots,x^{l_j,N}$ elements of a skew field.

We consider the matrices $A^{p_k}(x),$ which are defined by substituting the lines $l_1,\ldots, l_k$ of the identity matrix $I_{N}$ with the vectors $V_1(x),\ldots, V_k(x),$ respectively. Similarly we define the matrices  $A^{q_\lambda}(x).$  The matrices  $A^{p_k}(x)$ and $A^{q_\lambda}(x)$ participate in the following re-factorization problems
\begin{align} \label{refa_probs}
A^{p_k}(u)A^{q_\lambda}(v) = A^{q_\lambda}(y) A^{p_k}(x).
\end{align}

Unless otherwise stated, among the re-factorisation problems (\ref{refa_probs}) we concentrate to the ones that:
\begin{enumerate}
\item  $|p_k|=|q_{\lambda}|$ so $k=\lambda;$
\item $p_k\cap q_{\lambda}=\emptyset;$
\item $k=N/2,$ when $N$ is even, $k=(N-1)/2,$ when $N$ is odd.
\end{enumerate}
When the three conditions above are satisfied, without loss of generality we have
\begin{itemize}
\item  $p_k=p=\{1,\ldots, N/2\},$ $q_k=q=\{N/2+1,\ldots, N\},$ when $N$ is even,
\item $p_k=p=\{1,\ldots, \frac{N-1}{2}\},$ $q_k=q=\{\frac{N-1}{2}+1,\ldots, N-1\},$ when $N$ is odd.
\end{itemize}
Then we have the single  re-factorization problem
\begin{align} \label{refa_1}
A^{p}(u)A^{q}(v) = A^{q}(y) A^{p}(x),
\end{align}
that turns equivalent to the  map
$$
\phi^{p,q}:(x,y)\mapsto (u,v)= (u(x,y),v((x,y)).
$$
The following statements are in order.
\begin{enumerate}
\item The map  $\phi^{p,q}$ is a pentagon map.
\item  The inverse map $\left(\phi^{p,q}\right)^{-1},$ is equivalent to  $A^{q}(v)A^{p}(u) = A^{p}(x) A^{q}(y).
$
\item Mapping $\phi^{p,q}$ is birational.
\item For  the map $\phi^{q,p},$ there is
$
\phi^{q,p}=\tau\circ \left(\phi^{p,q}\right)^{-1}\circ \tau\equiv \left(\phi^{q,p}\right)^{op},
$
and it is equivalent to
\begin{align*}
A^{q}(u)A^{p}(v) = A^{p}(y) A^{q}(x).
\end{align*}
\item The re-factorization problem
\begin{align} \label{refa_2}
\left(A^{p}(u)\right)^{T}\left(A^{q}(v)\right)^{T} =\left(A^{q}(y)\right)^{T} \left(A^{p}(x)\right)^{T},
\end{align}
where the superscript $(T)$ stands for the transpose of a matrix, turns equivalent to  the reverse-pentagon   map
$$
\bar \phi^{p,q}:(x,y)\mapsto (u,v)= (u(x,y),v((x,y)).
$$
\end{enumerate}
The statements above can be easily proven. For instance let us sketch the proof of  statement $(1)$. For the proof we have to distinguish two cases, that is  when the order $N$ of the matrices   is even or odd. For even $N$ ($N=2k$), we have $p=\{1,\ldots,k\}$ and $q=\{k+1,\ldots,2k\}$ so
\begin{align*}
A^{\{1,\ldots,k\}}(x)=\left(\begin{array}{c|c}
 {\bf x}^{1}& {\bf x}^{2}   \\ \hline
{\bf 0}_k & {\bf I}_k
\end{array}\right),&& A^{\{k+1,\ldots,2k\}}(x)=\left(\begin{array}{c|c}
                           {\bf I}_k&{\bf 0}_k\\ \hline
                           {\bf x}^{1}& {\bf x}^{2}
                           \end{array}\right),
\end{align*}
where ${\bf 0}_k$ the order $k$ zero matrix,  ${\bf I}_k$ the order $k$ identity matrix and ${\bf x}^{i},$ $i=1,2,$ stands for  $k\times k$ block matrices. Then it is clear that the re-factorization problem
\begin{align*}
A^{\{1,\ldots,k\}}(u)A^{\{k+1,\ldots,2k\}}(v) = A^{\{k+1,\ldots,2k\}}(y) A^{\{1,\ldots,k\}}(x),
\end{align*}
 coincides with the $N=2$ re-factorization  problem
\begin{align} \label{refa:000}
A^{\{1\}}(u)A^{\{2\}}(v)=A^{\{2\}}(y)A^{\{1\}}(x),
 \end{align}
 when we consider the variables  $x^{1,j},y^{1,j},u^{1,j},v^{1,j},$ $j=1,2,$ that participate in (\ref{refa:000})   to be $k\times k$ matrices with entries in $A^{\times}.$ So what is valid for the order $N=2$ re-factorization problem (\ref{refa:000}), also holds for the re-factorization problem of order $N=2k.$ In the following Section we prove that (\ref{refa:000}) is equivalent to a pentagon map and that together with a similar analysis on odd-order re-factorization problems of the form (\ref{refa_1}), suffices as a proof of statement $(1)$.

Note that a consistency result such as the {\em 3-factorization property} \cite{Veselov:20031,koul-2009} that applies to matrices associated with Yang-Baxter maps, is not developed yet for the matrix factorizations associated with pentagon maps that we proposed above. Since  the {\em local $(N-1)-$gon equation} determines an $N-$gon map   \cite{Dimakis:2015,Hoissen2023}, by re-formulating the matrix factorizations presented above as  {\em local tetragon (2-gon) equations} we expect the anticipated consistency to follow.

\subsection{$N=2$ and the {\em  homogeneous normalization map}}
For $N=2,$ there is $p=\{1\}$ and  $q=\{2\},$ so we have the matrices
$$
A^{\{1\}}(x)=\begin{pmatrix}
 x^{1,1}& x^{1,2}   \\
0 & 1
\end{pmatrix}, \quad
A^{\{2\}}(x)= \begin{pmatrix}
 1&0\\
 x^{2,1}&  x^{2,2}
\end{pmatrix},
$$
and the re-factorization problem
$$
A^{\{1\}}(u)A^{\{2\}}(v)=A^{\{2\}}(y)A^{\{1\}}(x),
$$
that reads:
\begin{align} \label{refa0}
\begin{pmatrix}
u^{1,1}+u^{1,2}v^{2,1}&u^{1,2}v^{2,2}\\
v^{2,1}& v^{2,2}
\end{pmatrix}  =
\begin{pmatrix}
x^{1,1}&x^{1,2}\\
y^{2,1}x^{1,1}& y^{2,2}+y^{2,1}x^{1,2}
\end{pmatrix}.
\end{align}
 For simplicity, we denote $x^{(i)}:=x^{1,i},$ $u^{(i)}:=u^{1,i},$  $y^{(i)}:=y^{2,i}$ $v^{(i)}:=v^{2,i},$ $i=1,2$
and (\ref{refa0}) defines    uniquely the map $\phi^{\{1\},\{2\}}$ which we will denote for simplicity as $\phi.$
 The map $\phi,$ coincides (when all multiplications are taken in opposite order) with  the  {\em normalization map} that was introduced in \cite{Doliwa:2020}.  In detail, mapping $\phi,$ explicitly reads:
\begin{align} \label{norma:m00}
 \phi:\left(x^{(1)},x^{(2)};y^{(1)},y^{(2)}\right)\mapsto \left(u^{(1)},u^{(2)};v^{(1)},v^{(2)}\right),
\end{align}
 where:
\begin{align} \label{norma:m}
u^{(1)}=x^{(1)}-x^{(2)}\left(y^{(2)}+y^{(1)}x^{(2)}\right)^{-1}y^{(1)}x^{(1)},&&u^{(2)}=x^{(2)}\left(y^{(2)}+y^{(1)}x^{(2)}\right)^{-1},\\
v^{(1)}=y^{(1)}x^{(1)},&& v^{(2)}=y^{(2)}+y^{(1)}x^{(2)}.
\end{align}
Note that, mapping $\phi,$ admits the reduction $x^{(1)}+x^{(2)}=1,$  $y^{(1)}+y^{(2)}=1,$ and the reduced map coincides with the pentagon map obtained in  \cite{Kashaev:1998}.

\begin{prop} \label{prop2:1}
The map $\phi$ is a pentagon map.
\end{prop}

\begin{proof}
By direct computation  (see  the proof of Proposition \ref{prop3:2}, or  \cite{Doliwa:2014p,Doliwa:2020}) it can be  shown that  the map $\phi$ is a pentagon map i.e. it satisfies
$$
\phi_{12}\circ \phi_{13}\circ \phi_{23}=\phi_{23}\circ \phi_{12},
$$
where
$$
\phi_{12}=\phi\times id,\quad \phi_{23}=id\times \phi,\quad \phi_{13}=(id\times \tau)\circ \phi_{12}\circ (id\times \tau),
$$
and $\tau$ the transposition map i.e. $\tau: (x,y)\mapsto (y,x)$
\end{proof}

\subsection{$N=3,$  an extension of the homogeneous normalization map and non-constant Yang-Baxter and pentagon maps}

For $N=3$ we have the sets  $p=\{1\}$ and $q=\{2\}$ that correspond to the matrices
\begin{align*}
A^{\{1\}}(x)=\begin{pmatrix}
 x^{1,1}& x^{1,2}&x^{1,3}   \\
0 & 1 &0 \\
0&0&1
\end{pmatrix}, && A^{\{2\}}(x)=\begin{pmatrix}
1&0&0   \\
x^{2,1}& x^{2,2}&x^{2,3} \\
0&0&1
\end{pmatrix},
\end{align*}
Then the re-factorization problem
$$
A^{\{1\}}(u)A^{\{2\}}(v)=A^{\{2\}}(y)A^{\{1\}}(x),
$$
reads:
\begin{align} \label{refa1}
\begin{pmatrix}
u^{1,1}+u^{1,2}v^{2,1}&u^{1,2}v^{2,2}&u^{1,2}v^{2,3}+u^{1,3}\\
v^{2,1}& v^{2,2}&v^{2,3}\\
0&0&1
\end{pmatrix}  =
\begin{pmatrix}
x^{1,1}&x^{1,2}&x^{1,3}\\
y^{2,1}x^{1,1}& y^{2,2}+y^{2,1}x^{1,2}&y^{2,3}+y^{2,1}x^{1,3}\\
0&0&1
\end{pmatrix}.
\end{align}
We denote for simplicity $x^{(i)}:=x^{1,i},$ $u^{(i)}:=u^{1,i},$ $y^{(i)}:=y^{2,i}$ $v^{(i)}:=v^{2,i},$ $i=1,2,3$
and (\ref{refa1}) defines    uniquely the map
\begin{align} \label{New_Pent}
\phi^{\{1\},\{2\}}: \left(x^{(1)},x^{(2)},x^{(3)};y^{(1)},y^{(2)},y^{(3)}\right)\mapsto \left(u^{(1)},u^{(2)},u^{(3)};v^{(1)},v^{(2)},v^{(3)}\right),
\end{align}
where
\begin{align*}
u^{(1)}=x^{(1)}-x^{(2)}\left(y^{(2)}+y^{(1)}x^{(2)}\right)^{-1}y^{(1)}x^{(1)},&&u^{(2)}=x^{(2)}\left(y^{(2)}+y^{(1)}x^{(2)}\right)^{-1},\\
v^{(1)}=y^{(1)}x^{(1)},&& v^{(2)}=y^{(2)}+y^{(1)}x^{(2)},\\
u^{(3)}=x^{(3)}-x^{(2)}\left(y^{(2)}+y^{(1)}x^{(2)}\right)^{-1}\left(y^{(3)}+y^{(1)}x^{(3)}\right),&& v^{(3)}=y^{(3)}+y^{(1)}x^{(3)}.
\end{align*}

\begin{prop} \label{prop3:2}
The map $\phi^{\{1\},\{2\}}$ is a constant pentagon map.
\end{prop}
\begin{proof}
The proof that the map $\phi^{\{1\},\{2\}}$ is a  pentagon map follows by direct computations. Indeed, the left hand side of the constant pentagon equation
 \begin{align}  \label{pe_rela}
 \phi^{\{1\},\{2\}}_{12}\circ \phi^{\{1\},\{2\}}_{13}\circ \phi^{\{1\},\{2\}}_{23}=\phi^{\{1\},\{2\}}_{23}\circ \phi^{\{1\},\{2\}}_{12}
 \end{align}
 applied on $z^{(i)},$ $i=1,2,3,$ respectively gives
 \begin{align*}
z^{(1)}\xmapsto[]{\phi^{\{1\},\{2\}}_{12}}  z^{(1)}  \xmapsto[]{\phi^{\{1\},\{2\}}_{13}}  z^{(1)}x^{(1)}   \xmapsto[]{\phi^{\{1\},\{2\}}_{23}} z^{(1)}y^{(1)}x^{(1)} =\widehat z^{(1)},&&\\
z^{(2)}\xmapsto[]{\phi^{\{1\},\{2\}}_{12}}  z^{(2)}  \xmapsto[]{\phi^{\{1\},\{2\}}_{13}}  z^{(2)}+z^{(1)}x^{(2)}   \xmapsto[]{\phi^{\{1\},\{2\}}_{23}} z^{(2)}+z^{(1)}\left(y^{(2)}+y^{(1)}x^{(2)}\right) =\widehat z^{(2)},&&\\
z^{(3)}\xmapsto[]{\phi^{\{1\},\{2\}}_{12}}  z^{(3)}  \xmapsto[]{\phi^{\{1\},\{2\}}_{13}}  z^{(3)}+z^{(1)}x^{(3)}   \xmapsto[]{\phi^{\{1\},\{2\}}_{23}} z^{(3)}+z^{(1)}\left(y^{(3)}+y^{(1)}x^{(3)}\right) =\widehat z^{(2)},&&
\end{align*}
 which coincides with  the application of the right hand side of (\ref{pe_rela}) to the same components since
  \begin{align*}
z^{(1)}\xmapsto[]{\phi^{\{1\},\{2\}}_{23}} z^{(1)}y^{(1)} \xmapsto[]{\phi^{\{1\},\{2\}}_{12}}  z^{(1)}y^{(1)}x^{(1)} =\widehat z^{(1)},&&\\
z^{(2)}\xmapsto[]{\phi^{\{1\},\{2\}}_{23}} z^{(2)}+z^{(1)}y^{(2)} \xmapsto[]{\phi^{\{1\},\{2\}}_{12}}  z^{(2)}+z^{(1)}\left(y^{(2)}+y^{(1)}x^{(2)}\right) =\widehat z^{(2)},&&\\
z^{(3)}\xmapsto[]{\phi^{\{1\},\{2\}}_{23}} z^{(3)}+z^{(1)}y^{(3)} \xmapsto[]{\phi^{\{1\},\{2\}}_{12}}  z^{(3)}+z^{(1)}\left(y^{(3)}+y^{(1)}x^{(3)}\right) =\widehat z^{(3)}.&&
\end{align*}
 Similarly we can prove that the left hand side of (\ref{pe_rela}) applied on  $x^{(i)}, y^{(i)},$ $i=1,2,3,$ agrees with the application of the right hand side of (\ref{pe_rela}) to the same components.

To recapitulate, we provide the final expressions
\begin{align*}
(\widehat x, \widehat y, \widehat z)=\left\{\begin{aligned}
                      \phi^{\{1\},\{2\}}_{12}\circ \phi^{\{1\},\{2\}}_{13}\circ \phi^{\{1\},\{2\}}_{23}(x,y,z),\\
                      \phi^{\{1\},\{2\}}_{23}\circ \phi^{\{1\},\{2\}}_{12}(x,y,z),
                    \end{aligned}\right.
\end{align*}
where
\begin{align*}
\widehat x^{(1)}=x^{(2)}\left(y^{(2)}+y^{(1)}x^{(2)}\right)^{-1}y^{(2)}\left(x^{(2)}\right)^{-1}x^{(1)},&&\widehat z^{(1)}= z^{(1)}y^{(1)}x^{(1)},&&\\\\
\widehat x^{(2)}=x^{(2)}\left(y^{(2)}+y^{(1)}x^{(2)}\right)^{-1},&&\widehat z^{(2)}=z^{(2)}+z^{(1)}\left(y^{(2)}+y^{(1)}x^{(2)}\right),&&\\\\
\widehat x^{(3)}=x^{(2)}\left(y^{(2)}+y^{(1)}x^{(2)}\right)^{-1}\left(y^{(2)}\left(x^{(2)}\right)^{-1}x^{(3)}-y^{(3)}\right),&&\widehat z^{(3)}=z^{(3)}+z^{(1)}\left(y^{(3)}+y^{(1)}x^{(3)}\right),&&
\end{align*}
\begin{align*}
\widehat y^{(1)}= y^{(1)}x^{(1)}-\left(y^{(2)}+y^{(1)}x^{(2)}\right)\gamma^{-1}z^{(1)}y^{(1)}x^{(1)},&&\\
\widehat y^{(2)}=\left(y^{(2)}+y^{(1)}x^{(2)}\right)\gamma^{-1},&&\\
\widehat y^{(3)}=y^{(3)}+y^{(1)}x^{(3)}-\left(y^{(2)}+y^{(1)}x^{(2)}\right)\gamma^{-1}\delta,&&
\end{align*}
with $\gamma:=z^{(2)}+z^{(1)}\left(y^{(2)}+y^{(1)}x^{(2)}\right)$ and $\delta:=z^{(3)}+z^{(1)}\left(y^{(3)}+y^{(1)}x^{(3)}\right).$
\end{proof}
Note that the mapping $\phi^{\{1\},\{2\}}$ is of triangular form, with its components $x^{(j)},y^{(j)},$  $j=~1,2,$ varying as in the $N=2$ case, while the non-trivial coupling appears  in the $x^{(3)}$ and $y^{(3)}$ components. The associated  map with the components $x^{(j)},y^{(j)},$  $j=1,2,$ coincides with the normalization map (\ref{norma:m00}).



Furthermore, each one of the maps $\phi^{\{1\},\{2\}}, \phi^{\{1\},\{3\}}, \phi^{\{2\},\{3\}}$ is a constant pentagon map. These three pentagon maps  are related since it holds:
$$
\sigma_{23}\circ \phi^{\{1\},\{2\}}\circ (\sigma_{23})^{-1}=\phi^{\{1\},\{3\}},  \sigma_{12}\circ \phi^{\{1\},\{3\}}\circ (\sigma_{12})^{-1}=\phi^{\{2\},\{3\}},
$$
where
\begin{align*}
\sigma_{23}: (x^1,x^2,x^3;y^1,y^2,y^3)\mapsto (x^1,x^3,x^2;y^1,y^3,y^2),\\
 \sigma_{12}: (x^1,x^2,x^3;y^1,y^2,y^3)\mapsto (x^2,x^1,x^3;y^2,y^1,y^3).
\end{align*}

\begin{remark} \label{rem22}
Through direct computation, it can be shown that the inverse map   $\left(\phi^{\{1\},\{2\}}\right)^{-1}$ together with $\phi^{\{1\},\{2\}},$ do satisfy the constant ten-term relation (see Definition \ref{Def:ten-term}). Therefore, by following Proposition \ref{Def:Tetra-Penta},  the map $T:=\left(\phi^{\{1\},\{2\}}_{13}\right)^{-1}\circ \tau_{23}\circ \phi^{\{1\},\{2\}}_{13}$ is a tetrahedron map, whereas the map $P:=\left(\phi^{\{1\},\{2\}}_{13}\right)^{-1}\circ \tau_{01}\circ \tau_{23}\circ \phi^{\{1\},\{2\}}_{13}$ is a 4-simplex map. For their explicit form see Appendix \ref{app1}. For recent developments on 4-simplex maps, see \cite{Rizos:2022c}.
\end{remark}

 \begin{prop}
  The maps $\phi^{\{1\},\{2\}}, \phi^{\{1\},\{3\}},$ and $\phi^{\{2\},\{3\}}$ satisfy the following entwining (non-constant) Yang-Baxter equation
\begin{align}\label{entw_YB}
\phi^{\{1\},\{2\}}_{12}\circ \phi^{\{1\},\{3\}}_{13}\circ \phi^{\{2\},\{3\}}_{23}= \phi^{\{2\},\{3\}}_{23}\circ \phi^{\{1\},\{3\}}_{13}\circ \phi^{\{1\},\{2\}}_{12}.
\end{align}
\end{prop}

Note that the maps that correspond to sets of different cardinality are not constant pentagon maps i.e. the map $\phi^{\{1,2\},\{3\}}$ is not a constant pentagon map.
Nevertheless they participate in the solution of an entwining pentagon relation as the following proposition suggests.
\begin{prop} \label{fin_prop}
The maps $\phi^{\{1\},\{2,3\}}, \phi^{\{1\},\{2\}},$ and $\phi^{\{1,3\},\{2\}}$
\begin{enumerate}
\item are equivalent to the respective re-factorization problems
\begin{align} \label{refa3:1}
A^{\{1\}}(u)A^{\{2,3\}}(v)=A^{\{2,3\}}(y)A^{\{1\}}(x),\\ \label{refa3:2}
A^{\{1\}}(u)A^{\{2\}}(v)=A^{\{2\}}(y)A^{\{1\}}(x),
\end{align}
and
\begin{align} \label{refa3:3}
A^{\{1,3\}}(u)A^{\{2\}}(v)=A^{\{2\}}(y)A^{\{1,3\}}(x);
\end{align}
\item they satisfy the following entwining (non-constant) pentagon equation
\begin{align}\label{entw_Pe}
\phi^{\{1\},\{2,3\}}_{12}\circ \phi^{\{1\},\{2\}}_{13}\circ \phi^{\{1,3\},\{2\}}_{23}=\phi^{\{1,3\},\{2\}}_{23}\circ \phi^{\{1\},\{2,3\}}_{12}.
\end{align}
\end{enumerate}
\end{prop}

\begin{proof}
Let us prove the statements of this Proposition.
\begin{enumerate}
\item The matrices $A^{\{1\}},A^{\{2\}},A^{\{2,3\}}$ and $A^{\{1,3\}}$ that participate in (\ref{refa3:1})-(\ref{refa3:3}), respectively read
\begin{align*}
A^{\{1\}}(x)=\begin{pmatrix}
 x^{1,1}& x^{1,2}&x^{1,3}   \\
0 & 1 &0 \\
0&0&1
\end{pmatrix}, && A^{\{2\}}(x)=\begin{pmatrix}
1&0&0   \\
x^{2,1}& x^{2,2}&x^{2,3} \\
0&0&1
\end{pmatrix},\\
A^{\{2,3\}}(x)=\begin{pmatrix}
1&0&0   \\
x^{2,1}& x^{2,2}&x^{2,3} \\
x^{3,1}& x^{3,2}&x^{3,3}
\end{pmatrix},&& A^{\{1,3\}}(x)=\begin{pmatrix}
x^{1,1}& x^{1,2}&x^{1,3}    \\
0& 1&0 \\
x^{3,1}& x^{3,2}&x^{3,3}
\end{pmatrix}.
\end{align*}
The matrix re-factorization $A^{\{1\}}(u)A^{\{2,3\}}(v)=A^{\{2,3\}}(y)A^{\{1\}}(x),$ explicitly reads
\begin{align*}
\begin{multlined}[t]
\begin{pmatrix}
u^{(1)}+u^{(2)}v^{(1)}+u^{(3)}v^{(4)}&u^{(2)}v^{(2)}+u^{(3)}v^{(5)}&u^{(2)}v^{(3)}+u^{(3)}v^{(6)}\\
v^{(1)}&v^{(2)}&v^{(3)}\\
v^{(4)}&v^{(5)}&v^{(6)}
\end{pmatrix}\\
=\begin{pmatrix}
x^{(1)}&x^{(2)}&x^{(3)}\\
y^{(1)}x^{(1)}&y^{(2)}+y^{(1)}x^{(2)}&y^{(3)}+y^{(1)}x^{(3)}\\
y^{(4)}x^{(1)}&y^{(5)}+y^{(4)}x^{(2)}&y^{(6)}+y^{(4)}x^{(3)}
\end{pmatrix},
\end{multlined}
\end{align*}
and defines uniquely the mapping $\phi^{\{1\},\{2,3\}}$. Defining $x:=\left(x^{(1)},x^{(2)},x^{(3)},x^{(4)},x^{(5)},x^{(6)}\right),$ and similarly define $y,z,u,v,w,$   mapping $\phi^{\{1\},\{2,3\}}_{12}$ explicitly reads
\begin{align} \label{Pentagon12}
\phi^{\{1\},\{2,3\}}_{12}:(x,y,z)\mapsto (u,v,w),
\end{align}
where
\begin{align*}
\begin{aligned}
v^{(1)}=y^{(1)}x^{(1)},&&v^{(2)}=y^{(2)}+y^{(1)}x^{(2)},&&v^{(3)}=y^{(3)}+y^{(1)}x^{(3)},\\
v^{(4)}=y^{(4)}x^{(1)},&&v^{(5)}=y^{(5)}+y^{(4)}x^{(2)},&&v^{(6)}=y^{(6)}+y^{(4)}x^{(3)},
\end{aligned}\\
\begin{aligned}
u^{(2)}=\left(x^{(3)}\left(v^{(6)}\right)^{-1}-x^{(2)}\left(v^{(5)}\right)^{-1}\right)\left(v^{(3)}\left(v^{(6)}\right)^{-1}-v^{(2)}\left(v^{(5)}\right)^{-1}\right)^{-1}\\
u^{(3)}=\left(x^{(2)}\left(v^{(2)}\right)^{-1}-x^{(3)}\left(v^{(3)}\right)^{-1}\right)\left(v^{(5)}\left(v^{(2)}\right)^{-1}-v^{(6)}\left(v^{(3)}\right)^{-1}\right)^{-1},\\
u^{(1)}=x^{(1)}-u^{(2)}v^{(1)}-u^{(3)}v^{(4)},\\
u^{(j)}=x^{(j)},\quad j=4,5,6.
\end{aligned}
\end{align*}
In (\ref{refa1}) it is verified that the re-factorization problem $A^{\{1\}}(u)A^{\{2\}}(v)=A^{\{2\}}(y)A^{\{1\}}(x)$ is equivalent to the map $\phi^{\{1\},\{2\}}$  (\ref{New_Pent}).  Since we will make use of it later, in what  follows we present the map $\phi^{\{1\},\{2\}}_{13}$  that reads
\begin{align*}
\phi^{\{1\},\{2\}}_{13}:(x,y,z)\mapsto (u,y,w),
\end{align*}
where
\begin{align*}
\begin{aligned}
w^{(1)}=z^{(1)}x^{(1)},&&w^{(2)}=z^{(2)}+z^{(1)}x^{(2)},&&w^{(3)}=z^{(3)}+z^{(1)}x^{(3)},\\
 w^{(4)}=z^{(4)},&&w^{(5)}=z^{(5)},&&w^{(6)}=z^{(6)},\\
 u^{(1)}=x^{(1)}-x^{(2)}\left(w^{(2)}\right)^{-1}w^{(1)},&&u^{(2)}=x^{(2)}\left(w^{(2)}\right)^{-1},&&u^{(3)}=x^{(3)}-x^{(2)}\left(w^{(2)}\right)^{-1}w^{(3)},\\
  u^{(4)}=x^{(4)},&&u^{(5)}=x^{(5)},&&u^{(6)}=x^{(6)}.
\end{aligned}
\end{align*}
Working similarly, we can prove that (\ref{refa3:3}) is equivalent to the mapping $\phi^{\{1,3\},\{2\}}.$ Explicitly, mapping $\phi^{\{1,3\},\{2\}}_{23}$ reads
\begin{align*}
\phi^{\{1,3\},\{2\}}_{23}:(x,y,z)\mapsto (x,v,w)
\end{align*}
where
\begin{flalign} \label{Pentagon23}
\begin{aligned}
w^{(1)}=z^{(1)}y^{(1)}+z^{(3)}y^{(4)},&&w^{(4)}=z^{(4)},&&\\
w^{(2)}=z^{(2)}+z^{(1)}y^{(2)}+z^{(3)}y^{(5)},&&w^{(5)}=z^{(5)},&&\\
w^{(3)}=z^{(1)}y^{(3)}+z^{(3)}y^{(6)},&&w^{(6)}=z^{(6)},\\
v^{(1)}=y^{(1)}-y^{(2)}\left(w^{(2)}\right)^{-1}w^{(1)},&&v^{(4)}=y^{(4)}-y^{(5)}\left(w^{(2)}\right)^{-1}w^{(1)},&&{}\\
v^{(2)}=y^{(2)}\left(w^{(2)}\right)^{-1},&&v^{(5)}=y^{(5)}\left(w^{(2)}\right)^{-1},&&{}\\
v^{(3)}=y^{(3)}-y^{(2)}\left(w^{(2)}\right)^{-1}w^{(3)},&&v^{(6)}=y^{(6)}-y^{(5)}\left(w^{(2)}\right)^{-1}w^{(3)}.&&{}
\end{aligned}
\end{flalign}
\item The proof that the maps $\phi^{\{1\},\{2,3\}}, \phi^{\{1\},\{2\}},$ and $\phi^{\{1,3\},\{2\}}$ satisfy the non-constant pentagon equation follows by direct calculation. For example,  the left hand side of the non-constant pentagon equation  (\ref{entw_Pe})
     applied to $y^{(5)}$ gives
\begin{align*}
y^{(5)}\xmapsto[]{\phi^{\{1\},\{2,3\}}_{12}} y^{(5)}+y^{(4)}x^{(2)}  \begin{multlined}[t]  \xmapsto[]{\phi^{\{1\},\{2\}}_{13}}   y^{(5)}+y^{(4)}x^{(2)}\left(z^{(2)}+z^{(1)}x^{(2)}\right)^{-1}  \\ \xmapsto[]{\phi^{\{1,3\},\{2\}}_{23}}
\left(y^{(5)}+y^{(4)}x^{(2)}\right)\left(w^{(2)}+w^{(1)}x^{(2)}\right)^{-1}=\widehat y^{(5)},
\end{multlined}
\end{align*}
that coincides with the application of the right hand side of (\ref{entw_Pe})     to $y^{(5)}$  since
\begin{align*}
y^{(5)}\begin{multlined}[t] \xmapsto[]{\phi^{\{1,3\},\{2\}}_{23}} y^{(5)}\left(z^{(2)}+z^{(1)}y^{(2)}+z^{(3)}y^{(5)}\right)^{-1}\\ \xmapsto[]{\phi^{\{1\},\{2,3\}}_{12}} \left(y^{(5)}+y^{(4)}x^{(2)}\right)\left(w^{(2)}+w^{(1)}x^{(2)}\right)^{-1}=\widehat y^{(5)},
\end{multlined}
\end{align*}
where the expressions $w^{(1)}, w^{(2)}$ and $w^{(3)}$ are given in (\ref{Pentagon23}).

    In what follows we provide the final expressions
\begin{align*}
(\widehat x, \widehat y, \widehat z)=\left\{\begin{aligned}
                      \phi^{\{1\},\{2,3\}}_{12}\circ \phi^{\{1\},\{2\}}_{13}\circ \phi^{\{1,3\},\{2\}}_{23}(x,y,z),\\
                      \phi^{\{1,3\},\{2\}}_{23}\circ \phi^{\{1\},\{2,3\}}_{12}(x,y,z),
                    \end{aligned}\right.
\end{align*}
where
\begin{align*}
\widehat x^{(2)}=\begin{multlined}[t]\left(x^{(3)}\left(y^{(6)}+y^{(4)}x^{(3)}\right)^{-1}-x^{(2)}\left(y^{(5)}+y^{(4)}x^{(2)}\right)^{-1}\right)\\
\left(\left(y^{(3)}+y^{(1)}x^{(3)}\right)\left(y^{(6)}+y^{(4)}x^{(3)}\right)^{-1}-\left(y^{(2)}+y^{(1)}x^{(2)}\right)\left(y^{(5)}+y^{(4)}x^{(2)}\right)^{-1}\right)^{-1},
\end{multlined}&&\\
\widehat x^{(3)}=\begin{multlined}[t]\left(x^{(2)}\left(y^{(2)}+y^{(1)}x^{(2)}\right)^{-1}-x^{(3)}\left(y^{(3)}+y^{(1)}x^{(3)}\right)^{-1}\right)\\
\left(\left(y^{(5)}+y^{(4)}x^{(2)}\right)\left(y^{(2)}+y^{(1)}x^{(2)}\right)^{-1}-\left(y^{(6)}+y^{(4)}x^{(3)}\right)\left(y^{(3)}+y^{(1)}x^{(3)}\right)^{-1}\right)^{-1},
\end{multlined}&&\\
\widehat x^{(1)}=x^{(1)}-\widehat x^{(2)}y^{(1)}x^{(1)}-\widehat x^{(3)}y^{(4)}x^{(1)},&&\\
\widehat x^{(j)}=x^{(j)},\quad j=4,5,6,&& 
\end{align*}
\begin{align*}
\widehat y^{(1)}=y^{(1)}x^{(1)}-\left(y^{(2)}+y^{(1)}x^{(2)}\right)\left(x^{(2)}+\left(w^{(1)}\right)^{-1}w^{(2)}\right)^{-1}x^{(1)},&&\\
\widehat y^{(2)}=\left(y^{(2)}+y^{(1)}x^{(2)}\right)\left(w^{(2)}+w^{(1)}x^{(2)}\right)^{-1},&&\\
\widehat y^{(3)}=y^{(3)}+y^{(1)}x^{(3)}-\left(y^{(2)}+y^{(1)}x^{(2)}\right)\left(w^{(2)}+w^{(1)}x^{(2)}\right)^{-1}\left(w^{(3)}+w^{(1)}x^{(3)}\right),&&\\
\widehat y^{(4)}=y^{(4)}x^{(1)}-\left(y^{(5)}+y^{(4)}x^{(2)}\right)\left(w^{(2)}+w^{(1)}x^{(2)}\right)^{-1}w^{(1)}x^{(1)},&&\\
\widehat y^{(5)}=\left(y^{(5)}+y^{(4)}x^{(2)}\right)\left(w^{(2)}+w^{(1)}x^{(2)}\right)^{-1},&&\\
\widehat y^{(6)}=y^{(6)}+y^{(4)}x^{(3)}-\left(y^{(5)}+y^{(4)}x^{(2)}\right)\left(w^{(2)}+w^{(1)}x^{(2)}\right)^{-1}\left(w^{(3)}+w^{(1)}x^{(3)}\right),&&
\end{align*}
\begin{align*}
\widehat z^{(1)}=w^{(1)}x^{(1)},&&\\
\widehat z^{(2)}=w^{(2)}+w^{(1)}x^{(2)},&&\\
\widehat z^{(3)}=w^{(3)}+w^{(1)}x^{(3)},&&\\
\widehat z^{(j)}=z^{(j)},\quad j=4,5,6,&&
\end{align*}
where again $w^{(1)}, w^{(2)}$ and $w^{(3)}$ are given in (\ref{Pentagon23}).
\end{enumerate}
\end{proof}
Note that among the  mappings $\phi^{\{1\},\{2,3\}}, \phi^{\{1\},\{2\}},$ and $\phi^{\{1,3\},\{2\}},$  that participate in the entwining pentagon relation (\ref{entw_Pe}), only $\phi^{\{1\},\{2\}}$  is a constant pentagon map. As an outcome, mapping $\phi^{\{1\},\{2\}}$ is not $M\ddot{o}b$ equivalent neither to $\phi^{\{1\},\{2,3\}}$ nor to $\phi^{\{1,3\},\{2\}}.$ In addition, it can be checked that mapping $\Lambda:=\phi^{\{1\},\{2,3\}}\circ \left(\phi^{\{1\},\{2\}}\right)^{-1},$ is not  a boundary (see Definition \ref{boundary}) since it is not of the form $\lambda \times id_6,$ where $\lambda: \left(A^{\times}\right)^6\rightarrow \left(A^{\times}\right)^6,$ and $id_6$ the identity map acting on $\left(A^{\times}\right)^6.$ So mappings $\phi^{\{1\},\{2,3\}}, \phi^{\{1\},\{2\}},$ and $\phi^{\{1,3\},\{2\}},$ constitute a triple of {\em genuine} entwining pentagon maps.

As a final remark, the entwining pentagon maps of Proposition \ref{fin_prop}, as they stand do not satisfy the ten-term relation (\ref{ten-term}). In order the ten-term relation to be satisfied, it might be  that the introduction of  appropriate boundaries is necessary.

\section{Conclusions} \label{sec:4}
In this article we have proposed a specific class of matrices that their associated re-factorization problems turned equivalent either to pentagon,  reverse-pentagon, or Yang-Baxter maps. We  thoroughly  examined  the order $N=2$ and $N=3$ matrices and explicitly obtained  the associated maps.
For higher order matrices of the proposed form,  further studies are required.

 As an example, let us  briefly consider the case where the order of the matrices that participate in  re-factorization problems is $N=4.$ In this case, the maps
$\phi^{\{i_1,i_2\},\{j_1,j_2\}},$ where $i_1\neq i_2,$  $j_1\neq j_2,$ and $\{i_1,i_2\}\cap\{j_1,j_2\}=\emptyset,$ are constant pentagon maps and are equivalent to the map $\phi^{\{1,2\},\{3,4\}}.$ Mapping $\phi^{\{1,2\},\{3,4\}}$  can be obtained from  the homogeneous normalization map (\ref{norma:m00}) by considering $x^{(i)},y^{(i)},$ $i=1,2$ of (\ref{norma:m00}) to be $2\times 2$ matrices with entries in $A^{\times}.$ Note that the maps $\phi^{\{1\},\{2,3\}},$ $\phi^{\{2,3\},\{4\}},$ and $\phi^{\{1\},\{4\}}$ satisfy the following entwining Yang-Baxter equation
\begin{align*}
\phi^{\{1\},\{2,3\}}_{12}\circ \phi^{\{1\},\{4\}}_{13}\circ \phi^{\{2,3\},\{4\}}_{23}= \phi^{\{2,3\},\{4\}}_{23}\circ \phi^{\{1\},\{4\}}_{13} \circ \phi^{\{1\},\{2,3\}}_{12}.
\end{align*}
 Also from the maps $\phi^{\{i_1,i_2\},\{j_1,j_2\}},$ we can obtain as reductions the maps $\phi^{\{i\},\{j\}},$  $i\neq j\in \{1,\ldots,4\}$ which each one of them is a constant  pentagon map  and  equivalent (coincide up to conjugation) to the map $\phi^{\{1\},\{2\}}.$ In addition, by direct computation we can show that  the following entwining Yang-Baxter equations are satisfied
\begin{align*}
\phi^{\{i\},\{j\}}_{12}\circ \phi^{\{i\},\{k\}}_{13}\circ \phi^{\{j\},\{k\}}_{23}= \phi^{\{j\},\{k\}}_{23}\circ \phi^{\{i\},\{k\}}_{13}\circ \phi^{\{i\},\{j\}}_{12}, && i\neq j\neq k \neq i\in \{1,\ldots,4\},
\end{align*}
Furthermore, the  maps $\phi^{\{1\},\{2,3,4\}}$ and $\phi^{\{1,3,4\},\{2\}}$ are not pentagon but they participate in the following entwining pentagon equation
\begin{align} \label{ent_4_1}
\phi^{\{1\},\{2,3,4\}}_{12}\circ \phi^{\{1\},\{2\}}_{13}\circ \phi^{\{1,3,4\},\{2\}}_{23}=\phi^{\{1,3,4\},\{2\}}_{23}\circ \phi^{\{1\},\{2,3\}}_{12},
\end{align}
 while the non-pentagon map $\phi^{\{1\},\{2,3\}}$, together with the pentagon map $\phi^{\{1,4\},\{2,3\}},$  satisfy
\begin{align} \label{ent_4_2}
\phi^{\{1\},\{2,3\}}_{12}\circ \phi^{\{1\},\{2,3\}}_{13}\circ \phi^{\{1,4\},\{2,3\}}_{23}=\phi^{\{1,4\},\{2,3\}}_{23}\circ \phi^{\{1\},\{2,3\}}_{12}.
\end{align}
Note that (\ref{ent_4_1}), serves as the $N=4$ analogue of (\ref{entw_Pe}),  while the entwining pentagon equation (\ref{ent_4_2}) and the maps that satisfy it do not exist in the $N=3$ case.
Finally, we anticipate that the class of matrices introduced in this article, possibly with some modification, by participating  into the {\em local} pentagon equation \cite{Dimakis:2015,Dimakis:2018}, could determine novel hexagon maps i.e. maps that satisfy the {\em hexagon} equation \cite{Korepanov:2011}.

\section*{Acknowledgements}
\noindent The author would like to thank Igor G. Korepanov for the motivation to initiate with these studies and for the many valuable discussions on
pentagon, higher polygon and simplex equations. Furthermore the author  is thankful to Andrew Hone for the invitation and host at the  University of Kent, where part of these studies were initiated by the support of London Mathematical Society LMS grant ref. 42110 and by additional funding from the University of Kent.\\

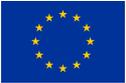
\parbox{.135\textwidth}{\begin{tikzpicture}[scale=.03]
\fill[fill={rgb,255:red,0;green,51;blue,153}] (-27,-18) rectangle (27,18);
\pgfmathsetmacro\inr{tan(36)/cos(18)}
\foreach \i in {0,1,...,11} {
\begin{scope}[shift={(30*\i:12)}]
\fill[fill={rgb,255:red,255;green,204;blue,0}] (90:2)
\foreach \x in {0,1,...,4} { -- (90+72*\x:2) -- (126+72*\x:\inr) };
\end{scope}}
\end{tikzpicture}} \parbox{.85\textwidth}
{This research is part of the project No. 2022/45/P/ST1/03998  co-funded by the National Science Centre and the European Union Framework Programme
 for Research and Innovation Horizon 2020 under the Marie Sklodowska-Curie grant agreement No. 945339. For the purpose of Open Access, the author has applied a CC-BY public copyright licence to any Author Accepted Manuscript (AAM) version arising from this submission.}

\appendix

\section{The explicit form of mappings $T$ and $P$ of Remark \ref{rem22}} \label{app1}

The Tetrahedron map $T$ introduced in Remark \ref{rem22} explicitly reads $T: (x,y,z)\mapsto (u,v,w),$ where
\begin{align*}
u^{(1)}=\left(z^{(1)}+z^{(2)}\left(x^{(2)}\right)^{-1}\right)^{-1}\left(y^{(1)}+z^{(2)}\left(x^{(2)}\right)^{-1}x^{(1)}\right),&&v^{(1)}=z^{(1)}x^{(1)},\\
u^{(2)}=\left(z^{(1)}+z^{(2)}\left(x^{(2)}\right)^{-1}\right)^{-1}y^{(2)},&&v^{(2)}=z^{(2)}+z^{(1)}x^{(2)},\\
u^{(3)}=\left(z^{(1)}+z^{(2)}\left(x^{(2)}\right)^{-1}\right)^{-1}\left(y^{(3)}+z^{(2)}\left(x^{(2)}\right)^{-1}x^{(3)}-z^{(3)}\right),
&&v^{(3)}=z^{(3)}+z^{(1)}x^{(3)},\\
w^{(1)}=y^{(1)}\left(y^{(1)}+z^{(2)}\left(x^{(2)}\right)^{-1}x^{(1)}\right)^{-1}\left(z^{(1)}+z^{(2)}\left(x^{(2)}\right)^{-1}\right),&&\\
w^{(2)}=\left(1+y^{(1)}\left(x^{(1)}\right)^{-1}x^{(2)}\left(z^{(2)}\right)^{-1}\right)^{-1}y^{(2)},&&\\
w^{(3)}=\left(x^{(1)}\left(y^{(1)}\right)^{-1}+x^{(2)}\left(z^{(2)}\right)^{-1}\right)^{-1}\left(z^{(3)}+x^{(1)}\left(y^{(1)}\right)^{-1}y^{(3)}-x^{(3)}\right).&&
\end{align*}

The 4-simplex map $P$ of Remark \ref{rem22} explicitly reads $P: (x,y,z,t)\mapsto (u,v,w,s),$ where
\begin{align*}
u^{(1)}=\left(1+y^{(2)}\left(t^{(2)}\right)^{-1}t^{(1)}\right)^{-1}y^{(1)},&&v^{(1)}=x^{(1)}+x^{(2)}z^{(1)},\\
u^{(2)}=y^{(2)}\left(t^{(2)}+t^{(1)}y^{(2)}\right)^{-1},&&v^{(2)}=x^{(2)}z^{(2)},\\
u^{(3)}=\left(t^{(2)}\left(y^{(2)}\right)^{-1}+t^{(1)}\right)^{-1}\left(t^{(2)}\left(y^{(2)}\right)^{-1}y^{(3)}-t^{(3)}\right),&&v^{(3)}=x^{(3)}+x^{(2)}z^{(3)},\\
s^{(1)}=z^{(1)}\left(x^{(1)}+x^{(2)}z^{(1)}\right)^{-1},&&w^{(1)}=t^{(1)}y^{(1)},\\
s^{(2)}=\left(1+z^{(1)}\left(x^{(1)}\right)^{-1}x^{(2)}\right)^{-1}z^{(2)},&&w^{(2)}=t^{(2)}+t^{(1)}y^{(2)},\\
s^{(3)}=\left(x^{(2)}+x^{(1)}\left(z^{(1)}\right)^{-1}\right)^{-1}\left(x^{(1)}\left(z^{(1)}\right)^{-1}z^{(3)}-x^{(3)}\right),&&w^{(3)}=t^{(3)}+t^{(1)}y^{(3)}.
\end{align*}





\begin{thebibliography}{10}
\expandafter\ifx\csname url\endcsname\relax
  \def\url#1{\texttt{#1}}\fi
\expandafter\ifx\csname urlprefix\endcsname\relax\def\urlprefix{URL }\fi
\expandafter\ifx\csname href\endcsname\relax
  \def\href#1#2{#2} \def\path#1{#1}\fi

\bibitem{Symes:1982}
W.~Symes, The {QR} algorithm and scattering for the finite nonperiodic {T}oda
  lattice, Physica D 4~(2) (1982) 275--280.

\bibitem{Deift:1989}
P.~Deift, L.~C. Li, C.~Tomei, Matrix factorizations and integrable systems,
  Commun. Pure Appl. Math. 42~(4) (1989) 443--521.

\bibitem{Moser:1991}
J.~Moser, A.~Veselov, Discrete versions of some classical integrable systems
  and factorization of matrix polynomials, Commun.Math. Phys. 139 (1991)
  217--243.

\bibitem{Nijhoff:1989}
F.~Nijhoff, V.~Papageorgiou, Lattice equations associated with the
  {L}andau-{L}ifschitz equations, Phys. Lett. A 141 (1989) 269--274.

\bibitem{Adler:1994}
V.~Adler, R.~Yamilov, Explicit auto-transformations of integrable chains, J.
  Phys. A: Math. Gen. 27~(2) (1994) 477.

\bibitem{Maillet:1989}
J.~Maillet, F.~Nijhoff, The tetrahedron equation and the four-simplex equation,
  Phys. Lett. A 134~(4) (1989) 221--228.

\bibitem{korepanov:1995}
I.~Korepanov, Algebraic integrable dynamical systems, $2+1$-dimensional models
  in wholly discrete space-time, and inhomogeneous models in 2-dimensional
  statistical physics, arXiv:9506003[solv-int] (1995).

\bibitem{Veselov:20031}
A.~Veselov, Yang-{B}axter maps and integrable dynamics, Phys. Lett. A 314
  (2003) 214--221.

\bibitem{Veselov:2003b}
Y.~B. Suris, A.~Veselov, Lax matrices for {Y}ang-{B}axter maps, J. Nonlin.
  Math. Phys 10~(2) (2003) 223--230.

\bibitem{koul-2009}
T.~Kouloukas, V.~Papageorgiou, Yang-{B}axter maps with first-degree polynomial
  $2\times 2$ lax matrices., J. Phys. Phys. 42 (2009) 404012.

\bibitem{Doliwa:2020}
A.~Doliwa, R.~Kashaev, Non-commutative birational maps satisfying
  {Z}amolodchikov equation, and {D}esargues lattices, J. Math. Phys. 61~(9)
  (2020) 092704.

\bibitem{Rizos:2022}
S.~Konstantinou-Rizos, Noncommutative solutions to {Z}amolodchikov's
  tetrahedron equation and matrix six-factorisation problems, Physica D 440
  (2022) 133466.

\bibitem{Dimakis:2015}
A.~Dimakis, F.~M{\"u}ller-Hoissen, Simplex and polygon equations, SIGMA 11 042
  (2015) 49pp.

\bibitem{korepanov:2022}
I.~Korepanov, Odd-gon relations and their cohomology, arXiv:2205.00405
  [math.QA] (2022).

\bibitem{MoorSeib:89}
G.~Moore, N.~Seiberg, Classical and quantum conformal field theory, Comm. Math.
  Phys. 123~(2) (1989) 177--254.

\bibitem{Maillet:1990}
J.~Maillet, Integrable systems and gauge theories, Nucl. Phys. (Proc. Suppl.)
  B18 (1990) 212--241.

\bibitem{Zakrzewski:1992}
S.~Zakrzewski, Poisson {L}ie groups and pentagonal transformations, Lett. Math.
  Phys. 24 (1992) 13--19.

\bibitem{Skandalis:1993}
S.~Baaj, G.~Skandalis, Unitaries multiplicatifs et dualit\'e pour le produits
  crois\'es de $c^*-$alg\'ebres, Ann. Sci. \'Ec. Norm. Sup. 26~(4) (1993)
  425--488.

\bibitem{Korepanov:2000}
I.~Korepanov, Invariants of {P}{L} manifolds from metrized simplicial
  complexes. three-dimensional case, J. Non. Math. Phys. 8~(2) (2001) 196--210.

\bibitem{Pachner:1991}
U.~Pachner, {P.}{L.} Homeomorphic manifolds are equivalent by elementary
  shellings, European J. Combin. 12~(2) (1991) 129--145.

\bibitem{Faddeev:1994}
L.~Faddeev, R.~Kashaev, Quantum dilogarithm, Mod. Phys. Lett. A 09~(05) (1994)
  427--434.

\bibitem{Kashaev:1999}
R.~Kashaev, On matrix generalizations of the dilogarithm, Theor. Math. Phys.
  118 (1998) 314--318.

\bibitem{Lobb:2009}
S.~Lobb, F.~Nijhoff, Lagrangian multiforms and multidimensional consistency, J.
  Phys. A: Math. Theor. 42~(45) (2009) 454013.

\bibitem{JNF}
N.~Hietarinta, J.~Joshi, F.~Nijhoff, Discrete Systems and Integrablity,
  Cambridge Texts in Applied Mathematics (No. 54), Cambridge University Press,
  2016.

\bibitem{Volkov:2011}
A.~Y. Volkov, {Pentagon identity revisited}, {I}{M}{R}{N} 2012~(20) (2011)
  4619--4624.

\bibitem{Fordy_Hone:2014}
A.~Fordy, A.~Hone, {Discrete integrable systems and {P}oisson algebras from
  cluster maps}, Comm. Math. Phys. 325 (2014) 527--584.

\bibitem{Kashaev:1998}
R.~Kashaev, S.~Sergeev, On pentagon, ten-term, and tetrahedron relations,
  Commun. Math. Phys. 195 (1998) 309--319.

\bibitem{Doliwa:2014p}
A.~Doliwa, S.~M. Sergeev, The pentagon relation and incidence geometry, J.
  Math. Phys. 55~(6) (2014) 063504.

\bibitem{Catino:2020}
F.~Catino, M.~Mazzotta, M.~Miccoli, Set-theoretical solutions of the pentagon
  equation on groups, Commun. Alg. 48~(1) (2020) 83--92.

\bibitem{Colazzo:2020}
I.~Colazzo, E.~Jespers, L.~Kubat, Set-theoretic solutions of the pentagon
  equation, Commun. Math. Phys. 380 (2020) 1003--1024.

\bibitem{Yang:1967}
C.~N. Yang, Some exact results for the many-body problem in one dimension with
  repulsive delta-function interaction, Phys. Rev. Lett. 19~(23) (1967)
  1312--1315.

\bibitem{Baxter:1982}
R.~Baxter, Exactly solved models in statistical mechanics, Academic Press,
  London, 1982.

\bibitem{Biedenharn_Louck:81}
L.~C. Biedenharn, J.~D. Louck, Angular momentum in quantum physics, Vol.~8 of
  Encyclopedia of Mathematics and its Applications, Addison-Wesley, Reading,
  MA, 1981.

\bibitem{Zamolodchikov:1980}
A.~B. Zamolodchikov, Tetrahedra equations and integrable systems in
  three-dimensional space, Sov. Phys. JETP 52~(2) (1980) 325--336.

\bibitem{Bazh_Stro:1982}
V.~V. Bazhanov, Y.~G. Stroganov, Conditions of commutativity of transfer
  matrices on a multidimensional lattice, Theoret. and Math. Phys. 52~(1)
  (1982) 685--691.

\bibitem{Kouloukas:2011}
T.~Kouloukas, V.~Papageorgiou, Entwining {Y}ang-{B}axter maps and integrable
  lattices, Banach Center Publ. 93 (2011) 163--175.

\bibitem{Kassotakis:2019}
P.~Kassotakis, Invariants in separated variables: {Y}ang-{B}axter, entwining
  and transfer maps, SIGMA 15~(048) (2019) 36pp.

\bibitem{Kels:2019II}
A.~Kels, Two-component {Y}ang-{B}axter maps and star-triangle relations,
  Physica D: Nonlinear Phenom. 448 (2023) 133723.

\bibitem{Adamopoulou:2021}
P.~Adamopoulou, S.~Konstantinou-Rizos, G.~Papamikos, Integrable extensions of
  the {A}dler map via {G}rassmann algebras, Theor. Math. Phys. 207 (2021)
  553--559.

\bibitem{Atkinson:2012i}
J.~Atkinson, Idempotent biquadratics, {Y}ang-{B}axter maps and birational
  representations of {C}oxeter groups, arXiv:nlin/1301.4613 (2013).

\bibitem{Atkinson:2014}
J.~Atkinson, Y.~Yamada, Quadrirational {Y}ang-{B}axter maps and the affine-{E}8
  painleve lattice, arXiv:nlin/1405.2745 (2014).

\bibitem{Mi+Pa+Wa:16}
A.~Mikhailov, G.~Papamikos, J.~Wang, {Darboux Transformation for the Vector
  Sine-Gordon Equation and Integrable Equations on a Sphere}, Lett. Math. Phys.
  106 (2016) 973--996.

\bibitem{Kass2}
P.~Kassotakis, M.~Nieszporski, V.~Papageorgiou, A.~Tongas, Integrable
  two-component systems of difference equations, Proc. R. Soc. A. 476 (2020)
  20190668.

\bibitem{Talalaev:2022}
V.~Bardakov, D.~Talalaev, Extensions of {Y}ang-{B}axter sets, arXiv:2206.09629
  [math.QA] (2022).

\bibitem{Caudrelier:2022}
V.~Caudrelier, A.~Gkogkou, B.~Prinari, Soliton interactions and {Y}ang-{B}axter
  maps for the complex coupled short-pulse equation, arXiv:2210.02265 [nlin.SI]
  (2022).

\bibitem{Doliwa_2014}
A.~Doliwa, Non-commutative rational {Y}ang-{B}axter maps, Lett. Math. Phys. 104
  (2014) 299--309.

\bibitem{Grahovski:2016}
G.~Grahovski, S.~Konstantinou-Rizos, A.~Mikhailov, Grassmann extensions of
  {Y}ang{\textendash}{B}axter maps, J. Phys. A: Math. Theor. 49~(14) (2016)
  145202.

\bibitem{Kassotakis:2:2021}
P.~Kassotakis, T.~Kouloukas, On non-abelian quadrirational {Y}ang-{B}axter
  maps, J. Phys. A: Math. Theor 55~(17) (2022) 175203.

\bibitem{Kassotakis:2022b}
P.~Kassotakis, Non-abelian hierarchies of compatible maps, associated
  integrable difference systems and {Y}ang-{B}axter maps, Nonlinearity 36~(5)
  (2023) 2514.

\bibitem{Kass1}
P.~Kassotakis, M.~Nieszporski, V.~Papageorgiou, A.~Tongas, Tetrahedron maps and
  symmetries of three dimensional integrable discrete equations., J. Math.
  Phys. 60 (2019) 123503.

\bibitem{Igonin:2022}
S.~Igonin, S.~Konstantinou-Rizos, Algebraic and differential-geometric
  constructions of set-theoretical solutions to the zamolodchikov tetrahedron
  equation, J. Phys. A: Math. Theor. 55~(40) (2022) 405205.

\bibitem{Hietarinta:1997}
J.~Hietarinta, Permutation-type solutions to the {Y}ang-{B}axter and other
  n-simplex equations, J. Phys. A: Math. Gen. 30~(13) (1997) 4757.

\bibitem{Dimakis_Korepanov:2021}
A.~Dimakis, I.~G. Korepanov, Grassmannian-parameterized solutions to direct-sum
  polygon and simplex equations, J. Math. Phys. 62~(5) (2021) 051701.

\bibitem{Hoissen2023}
F.~M{\"u}ller-Hoissen, On the structure of set-theoretic polygon equations,
  arXiv:2305.17974v2 [math-ph] (2023).

\bibitem{Rizos:2022c}
S.~Konstantinou-Rizos, Birational solutions to the set-theoretical 4-simplex
  equation, arXiv:2211.16338 [nlin.SI] (2022).

\bibitem{Dimakis:2018}
A.~Dimakis, F.~M{\"u}ller-Hoissen, Matrix {K}adomtsev-{P}etviashvili equation:
  Tropical limit, {Y}ang-{B}axter and {P}entagon maps, Theor. Math. Phys. 196
  (2018) 1164--1173.

\bibitem{Korepanov:2011}
I.~G. Korepanov, Relations in {G}rassmann algebra corresponding to three- and
  four-dimensional {P}achner moves, SIGMA 7 (2011) 117, 23~pages.

\end{thebibliography}


\end{document}